\documentclass[11pt]{article}

\usepackage[english]{babel}
\usepackage[ansinew]{inputenc}
\usepackage{amsmath, amsthm, amsfonts, amssymb}
\usepackage{latexsym}
\usepackage{bbm}
\usepackage{mathrsfs}
\usepackage{tikz}
\usepackage{graphicx}
\usepackage{multirow}
\usepackage{subfig}

\newtheorem{theorem}{Theorem}

\newtheorem{lemma}{Lemma}
\newtheorem{proposition}{Proposition}

\theoremstyle{definition}

\newtheorem{remark}{Remark}
\newtheorem{problem}{Problem}

\newtheorem{example}{Example}

\renewcommand{\qedsymbol}{$\Box$}

\newcommand{\N}{\mathbb{N}}
\newcommand{\Z}{\mathbb{Z}}
\newcommand{\Q}{\mathbb{Q}}

\newcommand{\C}{\mathbb{C}}

\newcommand{\e}{\varepsilon}
\newcommand{\la}{\lambda}

\newcommand{\Oh}{\mathcal{O}}

\newcommand{\Sat}{\mathrm{Sat}}

\newcommand{\Chow}{\mathit{Chow}}
\newcommand{\Det}{\mathit{Det}}  
 
\renewcommand{\det}{\mathrm{det}}

\newcommand{\Frac}{\mathrm{Frac}}
\newcommand{\ot}{\otimes}
\newcommand{\poly}{\mbox{\small poly}}
\newcommand{\dc}{\mathrm{dc}}
\newcommand{\VP}{\mathrm{VP}}
\newcommand{\VNP}{\mathrm{VNP}}
\newcommand{\Per}{\mathit{Per}}  
\newcommand{\sk}{\mathrm{sk}} 
\def\cq{/\hspace{-0.7ex}/}
\newcommand{\ol}[1]{\overline{#1}}

\DeclareMathOperator{\id}{id}
\DeclareMathOperator{\Gl}{GL}
\DeclareMathOperator{\Sl}{SL}

\DeclareMathOperator{\Sym}{Sym}

\DeclareMathOperator{\Pol}{Pol}

\DeclareMathOperator{\per}{per}

\begin{document}

\title{{\bf Permanent versus determinant: not via saturations}} 
\author{Peter B\"urgisser\thanks{
Institute of Mathematics, Technische Universit\"at Berlin, 
pbuerg@math.tu-berlin.de.
Partially supported by DFG grant BU 1371/3-2.}
\and 
Christian Ikenmeyer\thanks{
Texas A\&M University, 
ciken@math.tamu.edu}
 \and
Jesko H\"uttenhain\thanks{
Institute of Mathematics, Technische Universit\"at Berlin, 
jesko@math.tu-berlin.de.
Partially supported by DFG grant BU 1371/3-2.} 
}
\date{\today}

\maketitle

\begin{abstract}
Let $\Det_n$ denote the closure of the $\Gl_{n^2}(\C)$-orbit of the determinant polynomial $\det_n$ 
with respect to linear substitution. The highest weights (partitions) of irreducible $\Gl_{n^2}(\C)$-representations 
occurring in the coordinate ring of $\Det_n$ form a finitely generated monoid $S(\Det_n)$. 
We prove that the saturation of $S(\Det_n)$ contains all partitions $\la$ with length at most~$n$ 
and size divisible by~$n$. This implies that representation theoretic obstructions for the 
permanent versus determinant problem must be holes of the monoid $S(\Det_n)$. 
\end{abstract}

\smallskip


\smallskip


\section{Introduction}

It is a known fact~\cite{Val:79b} that every polynomial $f\in\C[X_1,\ldots,X_\ell]$ can be written 
as the determinant of some $n$ by $n$ matrix $A$, whose entries are affine linear 
polynomials in the variables~$X_i$. The smallest possible $n\ge 1$ is called the 
{\em determinantal complexity} $\dc(f)$ of $f$. 

The determinantal complexity of the permanent $\per_m$ of an $m$ by $m$ matrix $X_{ij}$ 
of variables is of paramount interest for complexity theory. 
It is easy to see that $\dc(\per_2) =2$. 
Only recently~\cite{abv:15}, it was shown that $\dc(\per_3)=7$, 
improving a similar result~\cite{hu-ik:14} in a restricted model. 
Grenet~\cite{grenet:11} showed that $\dc(\per_m) \le 2^m -1$. 
The best known lower bound, due to Mignon and Ressayre~\cite{mignon-ressayre:04}, 
states that $\dc(\per_m) \ge \frac12 m^2$. 
Over the reals, this lower bound was recently improved to 
$(m-1)^2+1$ in \cite{yabe:15}.

Valiant conjectured that $\dc(\per_m)$ grows superpolynomially fast in~$m$. 
The importance of this conjecture derives from the fact that it is equivalent to 
the separation of complexity classes $\VP_{\mathrm{ws}}\ne\VNP$, 
cf.~\cite{Val:79b,vali:82,buer:00-3,mapo:04}. 
The latter separation would constitute a substantial step towards resolving 
the famous $\mathrm{P}\ne\mathrm{NP}$ conjecture.

One can study the above questions as orbit closure problems in the following way, 
cf.\ Mumuley and Sohoni~\cite{gct1}. 
The group $G:=\Gl_{n^2}(\C)$ acts on the space $\Pol_n(\C^{n^2})$ of homogeneous forms of 
degree~$n$ by linear substitution, see \S\ref{se:proofs} for the precise definition. 
Let $\Det_n$ denote the closure of the orbit of $\det_n$. 
(One obtains the same closure for the Euclidean and the Zariski topology, 
cf.~\cite[Thm.~2.33]{mumf:95}.) 

Suppose $\dc(\per_m) \le n$. This means that $\per_m = \det(A)$ for some matrix $A$ of size $n\ge m$ 
with entries that are affine linear in the variables $X_{ij}$, $1\le i,j\le m$. 
Introducing an additional variable $Y$, substituting $X_{ij}$ by $X_{ij}/Y$, 
and multiplying with $Y^{n}$, we get 
$Y^{n-m} \per_m = \det(B)$ with a matrix $B$ of size $n$, whose entries are linear forms 
in $Y$ and the $X_{ij}$. We may view $\det(B)$ as being obtained from $\det_n$ by 
applying a (noninvertible) linear transformation of the variables. 
Since $\Gl_{n^2}$ is dense in $\C^{n^2\times n^2}$, we conclude that the 
{\em padded permanent} $Y^{n-m} \per_m$ is contained in the orbit closure~$\Det_n$.
Let $\Per_{n,m}$ denote the orbit closure of $Y^{n-m} \per_m$. 
Summarizing, we have seen that $\Per_{n,m}\not\subseteq \Det_n$ implies $\dc(\per_m) > n$. 

How could one possibly prove noncontainement of these orbit closures? 
Mulmuley and Sohoni~\cite{gct1,gct2} suggested the following elegant approach based on the representation theory 
of the group $G$. It is well known that the isomorphism types of irreducible rational $G$-modules are in bijective 
correspondence with the highest weights of $G$, which in turn can be encoded as vectors of integers $\la\in\Z^{n^2}$  
with weakly decreasing entries, cf.~\cite{fuha:91}. We write $\Lambda^+_G$ for the monoid of highest weights of $G$ 
and denote by $V_G(\lambda)$ an irreducible $G$-module of type~$\la$. 
Moreover, we denote by $\Lambda^+_G(\poly)$ the submonoid labeling the polynomial  $G$-modules.
Note that $\la \in \Lambda^+_G(\poly)$ iff all the components of $\la$ are nonnegative, i.e., 
$\la$ is a partition of $|\la| := \sum_i \la_i$. The length $\ell(\la)$ of $\la$ is the number of 
its nonzero parts.

Let $Z$ be any $G$-invariant closed subvariety of $\Pol_n(\C^{n^2})$. 
The group $G$ acts on the coordinate ring $\Oh(Z)$ of $Z$ via 
$(g F)(z):= F(g^{-1}z)$, where $g\in G, F\in\Oh(Z), z\in Z$. 
We are interested in the set of irreducible $G$-representations occurring in $\Oh(Z)$ 
and thus define 
\begin{equation}\label{def:S}
 S(Z) :=\big\{\la\in\Lambda^+_G \mid V_G(\la)^* \mbox{ occurs in }\Oh(Z)\big\} .
\end{equation}
It is known that $S(Z)$ is a finitely generated submonoid of $\Lambda^+_G(\poly)$, 
cf.~\cite{brion:87}. 
We call $S(Z)$ the {\em monoid of representations of the $G$-variety $Z$}.
We are mainly interested in the case where $Z$ is the orbit closure of some $h\in\Pol_n(\C^{n^2})$. 
Then it is easy to see that  $|\la| \equiv 0 \bmod n$ for any $\la\in S(Z)$. 
Moroever, if $h$ depends on $\ell$ variables only, then it is known that 
$\ell(\la) \le \ell$ for all $\la\in S(Z)$, cf.~\cite[\S6.3]{BLMW:11}.  

We make now the following general observation: let $Z'$ be a closed $G$-invariant subset of $Z$. 
The restriction from $Z$ to $Z'$ provides a surjective $G$-invariant ring morphism 
$\Oh(Z) \to\Oh(Z')$.  Hence Schur's lemma implies that $S(Z') \subseteq S(Z)$. 
These reasonings yield the implications 
\begin{equation}\label{eq:MS-attack}
  S(\Per_{n,m})\not\subseteq S(\Det_n) \Longrightarrow \Per_{n,m} \not\subseteq \Det_n 
  \Longrightarrow \dc(\per_m) > n .
\end{equation}
Mulmuley and Sohoni's~\cite{gct1,gct2} idea to attack the permanent versus determinant problem
is to exhibit {\em representation theoretic obstructions}, which are defined as highest weights 
$\la\in S(\Per_{n,m})\setminus S(\Det_n)$. 
Unfortunately, no such obstructions could be found so far.
It is even unclear whether they exist. 
The goal of this note is to shed more light on the principal difficulty of finding such obstructions.

If we want to prove at least $\dc(\per_m) > m^2 +1$ via \eqref{eq:MS-attack}, then 
we may assume that $n\ge m^2+1$. 
Since any $\la$ in $S(\Per_{n,m})$ satisfies $|\la| \equiv 0 \bmod n$ and 
$\ell(\la) \le m^2 +1 \le n$, 
we have to look for such partitions~$\la$ outside of $S(\Det_n)$. 

Before stating our main result, we need to introduce the concept of saturation, whose 
relevance for geometric complexity was already pointed out in~\cite{gct6}, 
see also~\cite{bor:09}.  

For the following compare~\cite[\S7.3]{MillerSturmfels}. 
Let $S$ be a submonoid of a free abelian group $F$ and $A(S):=S-S$ the group generated by~$S$. 
We call $S$ {\em saturated} if
$$ 
 \forall x \in A(S)\ \forall k\in \N_{>0} : kx \in S \Rightarrow x \in S .
$$ 
The {\em saturation} $\Sat(S)$ of $S$ is defined as the smallest saturated submonoid of~$F$ containing~$S$. 
It can also be characterized as the intersection of $A(S)$ with the rational cone generated by $S$. 
The elements in $\Sat(S)\setminus S$ are called the {\em holes} of $S$. 
The reason for this naming becomes apparent from a simple example: 
take $S=\N^2\setminus \{(0,1),(1,0)\}$, which has $\N^2$ as its saturation.
Replacing $S$ by $\Sat(S)$ means filling up the holes $(0,1),(1,0)$. 
Generally, understanding monoids is difficult due to the presence of holes. 

Our main result explains some of the difficulty in finding representation theoretic obstructions. 

\begin{theorem}\label{th:main}
The saturation of the semigroup of representations $S(\Det_n)$ contains 
$\{\la \in \Lambda^+_G(\poly) \mid \ell(\la) \le n, |\la| \equiv 0 \bmod n\}$, 
provided $n>2$. Hence representation theoretic obstructions must be holes of $S(\Det_n)$. 
\end{theorem}

\subsection{Related work}

Our work is closely related to a recent result by Shrawan Kumar~\cite{kumar:13}. 
A latin square of size~$n$ is an $n\times n$ matrix with entries from 
the set of symbols $\{1,\ldots,n\}$ such that in each row and in each column 
each symbol occurs exactly once. The sign of a latin square is defined as 
the product of the signs of all the row and column permutations. 
Depending on the sign, we can speak about even and odd latin squares. 
The Alon-Tarsi conjecture~\cite{AT:92} states that the number of even latin squares 
of size~$n$ is different from the number of odd  latin squares of size~$n$, provided $n$ is even.  
The Alon-Tarsi conjecture is known to be true if $n=p\pm 1$ where $p$ is a prime, 
cf.~\cite{Dri:98,Gly:10}.

\begin{theorem}[Kumar]\label{th:kumar-main}
Let $n$ be even. If the Alon-Tarsi conjecture for $n\times n$ latin squares holds, then 
$n\la \in S(\Det_n)$ for all partitions $\la$ such that $\ell(\la)\le n$. 
\end{theorem}

The above two theorems complement each other. 
Theorem~\ref{th:main} is unconditional and also provides information about the group 
generated by $S(\Det_n)$. 
Theorem~\ref{th:kumar-main} is conditional, but tells us about the stretching factor $n$. 
The proofs of both theorems focus on the orbit closure of the monomial $X_1\cdots X_n$,
called the $n$th Chow variety, 
but otherwise proceed differently. 
In fact, our proof gives information on the stretching factor in terms of certain degrees related 
to the normalization the Chow variety.  
Unfortunately, we were so far unable to bound the latter in a reasonable way. 
 
We conclude by briefly mentioning further related results. 
In~\cite{BLMW:11}, extending \cite{gct2}, it was shown that 
\begin{equation}\label{eq:S-So}
  S(\Det_n) \subseteq S^o (\Det_n) := \big\{ \la \in \Lambda^+_{\Gl_{n^2}}(\poly) \mid \exists d\ |\la| = dn, \
 \sk(\la,n\times d,n\times d) >0 \big\}  .
\end{equation}
Here $\sk$ denotes the {\em symmetric Kronecker coefficient} and 
$n\times d = (d,\ldots,d)$ is the rectangular partition of size $dn$. 
In fact, $S^o(\Det_n)$ encodes the polynomial irreducible representations 
occurring in coordinate ring of the orbit of $\det_n$. 
The interest in $S^o(\Det_n)$ comes from its close relationship to $S(\Det_n)$.
As a consequence of the fact that the $\Sl_{n^2}$-orbit of $\det_n$ is closed, 
for all $\la\in S^o(\Det_n)$, there exists $s\in\N$ such that 
$\la + s\e_n \in S(\Det_n)$, where $\e_n :=(1,\ldots,1)$ with $n^2$ ones, 
see \cite[Prop.~4.4.1]{BLMW:11}. 

In \cite{buci:09} it was shown that for all $\la\in\Lambda^+_{\Gl_{n^2}}(\poly)$ 
there exists $k\ge 1$ such that $k\la \in S^o(\Det_n)$. 
This was the first formal evidence that finding representation theoretic obstructions 
for the determinant versus permanent problem might be very hard. 
We conjecture that $S^o(\Det_n)$ generates the group 
$\{\la\in\Z^{n^2} \mid |\la| \equiv 0 \bmod n \}$ if $n>2$, cf.~\cite[Conjecture~9.2.3]{ike:12b}. 
Jointly with~\cite{buci:09}, this would imply that the saturation of $S^o(\Det_n)$ 
equals $\Lambda^+_{\Gl_{n^2}}(\poly)$.

The advantage of working with $S^o(\Det_n)$ is that we have an algorithm for computing 
symmetric Kronecker coefficients, which allows us to perform experiments.
Note that $\la\in S(\Per_{n,m})$ implies that $V_G(\la)^*$ occurs in $\Oh(\Pol_n (\C^{n^2}))$, 
which means that $V_G(\la)$ occurs in the plethysm $\Sym^d\Sym^n\C^{n^2}$, 
where $|\la|=dn$. 
It is also known that~\cite{kadish-landsberg:14}
\begin{equation}\label{eq:LaKa}
 \la\in S(\Per_{n,m})  \Longrightarrow \la_1 \ge (1-\frac{m}{n}) \, |\la| .
\end{equation}
(In fact, one may replace here $\Per_{n,m}$ by the orbit closure of any $Y^{n-m} h$, 
where $h$ is a form in $m^2$ variables.) 

\begin{problem}\label{prob:uno}
Are there partitions $\la$ with $\ell(\la)\le n$ and $|\la| = dn$ such that 
$V_G(\la)$ occurs in the plethysm $\Sym^d\Sym^n\C^{n^2}$ and 
$\sk(\la,n\times d,n\times d) =0$?
\end{problem}

\begin{remark}
The assumption on the plethysm is relevant. For instance, 
$\sk(\la,n\times d,n\times d) =0$, 
for the hook $\la=(dn-1,1)$, cf.~\cite[Lemma~3]{pak-panova:13}. 
But $V_G(\la)$ does not occur in $\Sym^d\Sym^n\C^{n^2}$:  
in fact, one can show that no hook $\la$ with at least two rows occurs there.
\end{remark}

A computer search performed by the second author (see ~\cite[Appendix A.1]{ike:12b})  
showed that there are no such partitions in the following cases: 
$n=2, d\le 20$, or $n=3,d\le 15$, or $n=4,d\le 10$.  

\medskip

\noindent{\bf Acknowledgements.} 
The idea for this note originated at the conference on Differential Geometry and its Applications in Brno (2013), 
where Joseph Landsberg organized a session on Geometric Complexity Theory. 
We are grateful to Shrawan Kumar, Joseph Landsberg, and Nicolas Ressayre for discussions and valuable suggestions. 
We also thank the Simons Institute for the Theory of Computing for hospitality and financial support during 
the program ``Algorithms and Complexity in Algebraic Geometry'', where this work was completed. 

\section{Preliminaries}

We consider here the following situation. 
$G:=\Gl_{n}(\C)$ is the general linear group and $U$ denotes its subgroup 
consisting of the upper triangular matrices. 
$W$ is a finite dimensional $\C$-vector space 
and a rational $G$-module such that scalar multiples 
of the unit matrix act on $W$ by nontrivial homotheties, i.e., there is a 
nonzero $a\in\Z$ such that $tI_n\cdot w = t^a w$ for $t\in\C^*,w\in W$. 
Further, $Z$ denotes a $G$-invariant, irreducible, locally closed nonempty subset of $W$. 
Then $Z$ is closed under multiplication with scalars in $\C^*$ 
by our assumption on the $G$-action.

We consider the induced action of~$G$ 
on the ring $\Oh(Z)$ of regular functions on $Z$ defined via 
$(g F)(z):= F(g^{-1}z)$, where $g\in G$, $F\in\Oh(Z)$, and $z\in Z$. 
As in \eqref{def:S} we define the 
{\em monoid $S(Z)$ of representations of the $G$-variety $Z$}.
It is known that $S(Z)$ is a finitely generated submonoid of $\Lambda^+_G$, 
cf.~\cite{brion:87}. We shall interpret $\Lambda^+_G$ as a subset of $\Z^n$ and 
denote by $A(S(Z))$ the group generated by $S(Z)$.
Moreover, we denote by $C_\Q(S(Z))$ the rational cone generated by $S(Z)$, that is, 
$$
 C_\Q(S(Z)) := \{ n^{-1} \la \mid n\in\N_{> 0}, \la \in S(Z) \} .
$$
It is easy to check that the saturation of $S(Z)$ is obtained as 
\begin{equation}\label{eq:Sat-AC}
 \Sat(S(Z)) = A(S(Z)) \cap C_\Q(S(Z)). 
\end{equation}
We denote by $\Frac(R)$ the field of fractions of an 
integral ring $R$. 
We have an induced $G$-action on the field of fractions $\C(Z):=\Frac(\Oh(Z))$ 
and denote by $\C(Z)^U$ its subfield of $U$-invariants. 
Recall that a highest weight vector is a $U$-invariant weight vector.
The following lemma is well known, but we include its proof for 
completeness.  

\begin{lemma}\label{le:FracU}
We have $\Frac (\Oh(Z)^U) = \C(Z)^U$. 
Moreover, for a highest weight vector $f\in \C(Z)^U$, 
there exist highest weight vectors 
$p,q\in\Oh(Z)^U$ such that $f=p/q$.
\end{lemma}

\begin{proof} 
The inclusion $\Frac (\Oh(Z)^U) \subseteq\C(Z)^U$ is obvious. 
Now let $f\in \C(Z)^U$ and consider the ideal
$J := \{ q \in \Oh(Z) \mid qf \in \Oh(Z) \}$ of $\Oh(Z)$. 
Since $J\ne 0$ we have $J^U\ne 0$, cf.~\cite[\S17.5]{humphreys:75}.
Choose a nonzero $q \in J^U$. Then 
$p:= qf \in \Oh(Z)^U$ and $ f=p/q$, 
hence $f\in\Frac(\Oh(Z)^U)$. 

If, additionally, $f\in \C(Z)^U$ is a weight vector, 
we can argue as before, choosing $q\in J^U$ as a highest weight vector.
Then $p:= qf$ is a highest weight vector in $\Oh(Z)$.
The assertion follows. 
\end{proof}

\begin{remark}\label{re:o-subset}
If $Z_o$ is a nonempty $G$-invariant open subset  of $Z$, then
$S(Z) \subseteq S(Z_o)$ and $A(S(Z)) = A(S(Z_o))$. 
This follows immediately from Lemma~\ref{le:FracU}. 
\end{remark}

%

Suppose now that $Z$ is a closed subset of $W$, hence an affine variety. 
Then we have an induced $G$-action on the normalization $\tilde{Z}$ and 
the canonical map $\pi\colon\tilde{Z} \to Z$ is $G$-invariant. 
Indeed, the integral closure $R$ of $\Oh(Z)$ in $\C(Z)$ is $G$-invariant 
and $\pi\colon\tilde{Z} \to Z$ is defined by $\Oh(Z)\hookrightarrow R$.
By construction, we can identify $\C(\tilde{Z})$ with $\C(Z)$. 
Note that $S(Z)\subseteq S(\tilde{Z})$ since $\pi$ is surjective.

\begin{proposition}\label{pro:Ssat-norm}
We have $A(S(\tilde{Z})) = A(S(Z))$ and $\Sat(S(\tilde{Z})) = \Sat(S(Z))$.
More precisely, assume that $\Oh(\tilde{Z})$ is generated as an $\Oh(Z)$-module by $e$ elements. Then for all $\la\in S(\tilde{Z})$, there is some $i<e$ such that $(e-i)\cdot\la \in S(Z)$.
\end{proposition}

\begin{proof}
Let $\la\in S(\tilde{Z})$ and 
$f\in\Oh(\tilde{Z})^U$ be a highest weight vector of weight~$\la$. 
Then $f\in\C(Z)^U$ and Lemma~\ref{le:FracU} shows the existence of 
highest weight vectors $p,q\in\Oh(Z)^U$, say with the weights $\mu,\nu \in S(Z)$, 
respectively, such that $f=p/q$.
Therefore $\la =\mu - \nu \in A(Z)$.
This shows the equality for the groups.

Due to \eqref{eq:Sat-AC}, it suffices to prove that  
$C_\Q(S(Z))=C_\Q(S(\tilde{Z}))$.
 Suppose $f\in\Oh(\tilde{Z})$ is a highest weight vector of weight $\la$.
Since $f$ is integral over $\Oh(Z)$, there are $e\in\N_{\ge 1}$ and 
$a_0,\ldots,a_{e-1}\in\Oh(Z)$ such that
such that 
\begin{equation}\label{eq:norm-equ}
  \mbox{ $f^e + \sum_{i=0}^{e-1} a_{i} f^{i} = 0$.}
\end{equation}
We assume that the degree $e$ is the smallest possible.

Note that $e$ is at most the size of a generating set of $\Oh(\tilde{Z})$ as an $\Oh(Z)$-module, as follows from the classical theory of integral extensions, see \cite[Prop.~5.1 and the proof of Prop.~2.4]{am:69}.

Consider the weight decomposition 
$a_i = \sum_\mu a_{i,\mu}$ 
of $a_i$, where $a_{i,\mu}$ has the weight~$\mu$. 
Then $a_{i,\mu} f^i$ has the weight $\mu + i\la$. 
Moreover,  $f^e$ has the weight $e\la$. 
Since the component of weight $e\la$ in \eqref{eq:norm-equ} must vanish, 
we have 
$$
 \mbox{$f^e + \sum_{i=0}^{e-1} a_{i,(e-i)\la} f^{i} = 0$.} 
$$
Since the degree $e$ is the smallest possible, 
the above is the minimal polynomial of $f$. 
Applying $u\in U$ to  the above equation and using $u f=f$, we get 
$
 f^e + \sum_{i=0}^{e-1} u a_{i,(e-i)\la} f^{i}  = 0 .
$
The uniqueness of the minimal polynomial implies that $u a_{i,(e-i)\la} = a_{i,(e-i)\la}$ for all~$i$. 
Hence $a_{i,(e-i)\la}$ is a highest weight vector, provided it is nonzero. 
Since there exists~$i<e$ with $a_{i,(e-i)\la}\ne 0$, we see that 
$(e-i)\la \in S(\Oh(Z))$ for this~$i$. 
We conclude that $\la\in C_\Q(\Oh(Z))$.
\end{proof}

%
%
%
%
%
%

\begin{example} 
If $G=(\C^*)^n$ we can identify $\Lambda^+_G$ with $\Z^n$. 
Let $S\subseteq\Z^n$ be a finitely generated submonoid and 
consider the finitely generated subalgebra 
$\C[S] := \oplus_{s\in S} \, \C X_1^{s_1}\cdots X_n^{s_n}$ 
of $\C[X_1^{\pm 1},\ldots,X_n^{\pm 1}]$.  
If we interpret $\C[S]$ as the coordinate ring of an affine variety $Z$, 
then we have a $(\C^*)^n$-action on $Z$ and 
$S(Z)$ can be identified with $S$
($Z$ is called  a toric variety).  
It is known that $\C[\Sat(S)]$ equals the integral closure of $\C[S]$ 
in $\Frac(\C[S])$, cf.~\cite[Prop.~7.25, p.~140]{MillerSturmfels}. 
Thus the affine variety corresponding to $\Sat(S)$ equals the 
normalization of $Z$. This illustrates 
Proposition~\ref{pro:Ssat-norm}
in the special case of toric varieties.
\end{example}

\section{Proofs}
\label{se:proofs}

Write $V:=\C^n$, $G:=\Gl(V)$ and consider $W := \Sym^n V$ with is natural $G$-action.
We can interpret $W$ as the space $\Pol_n V^*$ of degree~$n$ forms 
on $V^*$. Note that when identifying $V=\C^n$ with $V^*$, the following 
$G$-action on $\Pol_n (\C^n)$ results: 
$(gw)(x):=w(g^T x)$, where $g\in G, w\in \Pol_n (\C^n), x\in \C^n$. 
(This point is sometimes confusing.)  
This is the action considered in the introduction.
Note that 
$\Oh(W)_d \simeq \Sym^d \Sym^n V^*$
as $G$-modules. 

Let $X_1,\ldots,X_n$ be a basis of $V$ and consider $w_n := X_1\cdots X_n\in W$.  
Clearly, the orbit of~$w_n$ consists of all symmetric products $v_1\cdots v_n$ of $n$ linearly independent vectors~$v_i$.  
We define  
$$
 \Chow_n :=  \big\{v_1\cdots v_n \mid v_1,\ldots,v_n \in V \big\} \subseteq W .
$$
$\Chow_n$ is a special case of a 
{\em Chow variety}, see \cite{GKZ:94}. 
We note that $\Chow_2 = \Sym^2\C^2$.  

\begin{lemma}\label{le:clo-chow}
$\Chow_n$ equals the $G$-orbit closure of $w_n$. 
\end{lemma}


\begin{proof}
We have $\Chow_n\subseteq \ol{G w_n}$ since $G$ is dense in $\C^{n\times n}$. 
For the converse inclusion, suppose the nonzero $f\in \ol{G w_n}$ is the limit of 
a sequence $(c_k v_1^{(k)}\cdots v_n^{(k)})$, where $c_k\in\C^*$ and $v_i^{(k)} \in V$ 
with $\|v_i^{(k)}\| =1$. By compactness of the unit sphere in $V$ we may assume that, 
after going over to a subsequence, that $(v_i^{(k)})$ is convergent, for $i=1,\ldots,n$. 
Say, $\tilde{v}_i = \lim_{k\to\infty} v_i^{(k)}$. Then $(v_1^{(k)}\cdots v_n^{(k)})$ converges 
to $\tilde{v}_1\cdots \tilde{v}_n$, which is nonzero. It follows easily that $(c_k)$ converges 
to some $c$ and $f=c\tilde{v}_1\cdots \tilde{v}_n \in \Chow_n$.
\end{proof}

When we identify $X_i$ with the variable $X_{ii}$, then $\Chow_n$ is
contained in the $\Gl_{n^2}$-orbit closure $\Det_n$.  
Indeed, let $\varepsilon\in\C^*$. The linear substitution 
$X_{ii} \mapsto X_{ii}$, $X_{ij} \mapsto \varepsilon X_{ij}$ for $i\ne j$, maps $\det_n$ to 
$X_{11}\cdots X_{nn} + \varepsilon p$ for some polynomial $p$. We obtain 
$X_{11}\cdots X_{nn}$ in the limit when $\varepsilon$ goes to zero.

The basic strategy, as in Kumar~\cite{kumar:13}, is to replace $\Det_n$ by the considerably simpler $\Chow_n$ and 
to exhibit elements in the monoid of representations of the latter. More specifically, 
we have $S(\Chow_n)\subseteq S(\Det_n)$ and hence $\Sat(S(\Chow_n))\subseteq \Sat(S(\Det_n))$.
Our main Theorem~\ref{th:main} is an immediate consequence of the following result.

\begin{theorem}\label{th:main-Chow}
We have 
$\Sat(S(\Chow_n)) = \big\{ \la\in\Lambda^+_{\Gl_n}(\poly) \mid |\la| \equiv 0 \bmod n \big\}$, provided $n>2$. 
\end{theorem}

According to Proposition~\ref{pro:Ssat-norm}, for proving this, we may replace $\Chow_n$ by its normalization. 
It is crucial that the latter has an explicit description, that we describe next. 

For the following arguments compare~\cite{brion:93} and \cite{lands-survey:13}. 
The symmetric group $S_n$ operates on the group 
$T_n:=\{(t_1,\ldots,t_n)\in \C^n \mid t_1\cdots t_n =1 \}$ 
by permutation. The corresponding semidirect product $H_n :=T_n \rtimes S_n$ acts on $V^{n}$ 
by scaling and permutation. Note that this action commutes with the $G$-action. 
Consider the product map 
$$
 \varphi_n\colon V^{n} \to \Chow_n,\  (v_1,\ldots,v_n) \mapsto v_1\cdots v_n ,
$$
which is surjective and $G$-equivariant. 
Clearly, $\varphi_n$ is invariant on $H_n$-orbits. 
Moreover, the fiber of a nonzero $w\in\Chow_n$ is an $H_n$-orbit.
This easily follows from the uniqueness of polynomial factorization (interpreting the $v_i$ 
as linear forms). 

The group $G$ contains the subgroup $\{ t\cdot\id_V \mid t \in \C^\times \}\cong \C^\times$ and this $\C^\times$-action induces a natural grading on the coordinate rings $\Oh(V^n)$ and $\Oh(\Chow_n)$. Since $\varphi_n$ is $G$-equivariant, the corresponding comorphism $\varphi_n^\ast\colon\Oh(\Chow_n)\to\Oh(V^n)$ is in particular a homomorphism of graded $\C$-algebras. However, the $\C^\times$-action on $\Oh(\Chow_n)$ is not the canonical one because $t\in\C^\times$ acts by multiplication with the scalar $t^n$. A homogeneous element of degree $kn$ in $\Oh(\Chow_n)$ is the restriction of a $k$-form on $W$.

The categorical quotient $V^{n} \cq H_n$ is defined as the affine variety that has as 
its coordinate ring the ring of $H_n$-invariants $\Oh(V^n)^{H_n}$, 
which is finitely generated and graded since $H_n$ is reductive, cf.~\cite{kraf:84}.
The inclusion $\Oh(V^n)^{H_n}\hookrightarrow \Oh(V^n)$ defines a 
$G$-equivariant, surjective morphism 
$\pi\colon V^n \to V^{n} \cq H_n$. 
Since $V^{n}$ is normal, the quotient $V^{n} \cq H_n$ is normal as well, 
see~\cite[p.~45]{dolgachev:03} for the easy proof. 
The map $\varphi_n$ factors through a $G$-equivariant morphism 
\begin{equation}\label{eq:normal-chow}
 \psi_n\colon V^{n} \cq H_n \to \Chow_n ,
\end{equation}
due to the universal property of categorical quotients. Moreover, by construction, the fibers of~$\psi_n$ over a nonzero $w\in \Chow_n$ consist of just one element. The action of $H_n$ on $\Oh(V^n)$ is linear, therefore it respects the grading. Hence, the comorphism $\psi_n^\ast\colon\Oh(\Chow_n)\to\Oh(V^n)^{H_n}$ is again a homomorphism of graded $\C$-algebras. 

The following is shown in~\cite[Prop., p.~351]{brion:93}. 

\begin{lemma}\label{le:psi-finite} 
The morphism $\psi_n$ is finite. \qedsymbol 
\end{lemma}

We conclude that $\psi_n\colon V^{n} \cq H_n \to \Chow_n$ is the normalization of $\Chow_n$,  
since $V^{n} \cq H_n$ is normal and $\psi_n$ is finite and generically injective, 
cf.~\cite[\S5.2]{shaf:94}. 

\begin{lemma}\label{le:CR-Chow}
We have 
$$
 S(V^n\cq H_n)= \big\{ \la\in\Lambda^+_G(\poly) \mid 
\mbox{$n$ divides $|\la|$ and $V_G(\la)$ occurs in $\Sym^n\Sym^{\frac{|\la|}{n}} \C^n$ \big\}  }.
$$
\end{lemma}

\begin{proof} 
We shall decompose the coordinate ring $\Oh(V^n\cq H_n)$ with respect to the $G$-action.
We have $\Oh(V) = \oplus_{k\in\N} \Sym^k V^*$ and therefore
$$
  \Oh(V^{n}) = \Oh(V)^{\ot n} = 
  \bigoplus_{k_1,\ldots,k_n\in\N} \Sym^{k_1} V^*\ot \cdots \ot \Sym^{k_n} V^*. 
$$
Taking $T_n$-invariants yields 
$$
  \Oh(V^{n})^{T_n} = 
  \bigoplus_{k\in\N} \Sym^{k} V^* \ot \cdots \ot \Sym^{k} V^* . 
$$
Taking $S_n$-invariants gives 
\begin{equation}
  \label{eq:chow-grading-decomposition} 
  \Oh(V^n\cq H_n) \simeq \Oh(V^{n})^{T_n\rtimes S_n} = 
  \bigoplus_{k\in\N} \Sym^n \Sym^{k} V^* .
\end{equation} 
and the assertion follows.  
\end{proof}

Recalling Theorem~\ref{th:kumar-main}, we expect the stretching factor to be $n$.
Without relying on the Alon-Tarsi conjecture however, the following exponential bound is the best we can currently provide.

\begin{proposition}\label{cor:SChow}
$S(\Chow_n)$  generates the rational cone 
$\{q\in\Q^n\mid q_1\ge\cdots\ge q_n\ge 0\}$. 

More precisely: Assume $n>2$. For each partition $\la$ with $\ell(\la)\le n$ and $|\la|\in n\N$, there is some number $N < n^{n^2-2n}$ such that we have $2N\cdot\la \in S(\Chow_n)$.
\end{proposition}

\begin{proof}
For the first statement, it is sufficient to show that $C_\Q(S(\Chow_n))$ contains any partition 
$\la$ with $|\la| = nk$ and $\ell(\la) \le n$. 
According to Proposition~\ref{pro:Ssat-norm}, the semigroups 
$S(\Chow_n)$ and $S(V^n\cq H_n)$ generate the same rational cone.
So we need to show that $\la$ lies $C_\Q(S(V^n\cq H_n))$. 

In~\cite{bci:10} (see \cite{ike:12b} for a simpler proof) 
it was shown that 
$V_{\Gl_n}(2\la)$ occurs in $\Sym^n\Sym^{2k} (\C^n)$. 
Thus Lemma~\ref{le:CR-Chow} implies that 
$\la$ lies in the rational cone generated by $S(V^n\cq H_n)$. 

We will now make the above reference to Proposition~\ref{pro:Ssat-norm} precise. Recall that the comorphism of $\psi_n$ from \eqref{eq:normal-chow} is an integral extension $\psi_n^\ast\colon\Oh(\Chow_n)\hookrightarrow \Oh(V^n)^{H_n}$ of graded $\C$-algebras. Note that the grading induced by $G$ on these two rings is such that only degrees which are multiples of $n$ contain nonzero elements. We therefore change the grading such that degree $n\cdot k$ becomes degree $k$. This means that the the direct sum \eqref{eq:chow-grading-decomposition} is the grading of $\Oh(V^n)^{H_n}$ and the grading of $\Oh(\Chow_n)$ is the grading induced by the canonical grading induced by the polynomial ring $\Oh(W)$. Therefore, $\Oh(\Chow_n)$ is generated by elements of degree one.

By \cite[Lemma~2.4.7]{dk:02}, we may choose a system of parameters $Y_0,\ldots,Y_r\in\Oh(\Chow_n)_1$ from among sufficiently generic linear forms. This means that $R:=\C[Y_0,\ldots,Y_r]$ is a polynomial ring in the $Y_i$ and $\Oh(\Chow_n)$ is a finite $R$-module. 
It follows that $\Oh(V^n)^{H_n}$ is integral over $R$ and more precisely, $Y_0,\ldots,Y_r$ is also a homogeneous system of parameters for $\Oh(V^n)^{H_n}$. 

We note that $r+1=\dim(V^n\cq H_n)=\dim(V^n)-\dim(H_n)=n^2 - n + 1$, so $r=n^2-n$. 
Furthermore, $\Oh(V^n)^{H_n}$ is Cohen-Macaulay \cite[Prop.~2.5.5]{dk:02} and by \cite[Prop.~2.5.3]{dk:02}, it follows that $\Oh(V^n)^{H_n}$ is free as an $R$-module, i.e., $\Oh(V^n)^{H_n}\cong R^D$ for some $D\in\N$. Since $D$ is the number of generators of $\Oh(V^n)^{H_n}$ as an $R$-module, $D$ is a (possibly rough) upper bound for the number of generators of $\Oh(V^n)^{H_n}$ as an $\Oh(\Chow_n)$-module. The second assertion follows from Proposition~\ref{pro:Ssat-norm} as soon as we have verified that $D<n^{n^2-2n}$.

The Hilbert polynomial of $R^D$ is a polynomial of degree $r$ whose leading coefficient is equal to $\frac{D}{r!}$. Since \eqref{eq:chow-grading-decomposition} gives the grading of $\Oh(V^n)^{H_n}$, the Hilbert polynomial of $\Oh(V^n)^{H_n}$ is given, for sufficiently large $k$, by the map 
\[ k \mapsto \binom{\binom{k+n-1}{n-1}+n-1}{n}, \]
whose leading coefficient is $\frac{1}{n!(n-1)!^n}$. Since $r=n^2-n$, we have $D =  \frac{(n^2-n)!}{n!(n-1)!^n}$. We apply Stirling's approximation:
\[ \forall n>0\colon\quad 1 \le \frac{n!}{\sqrt{2\pi} \cdot e^{-n} \cdot n^{n+\frac12}} \le e^{\frac{1}{12n}} \]
to the fraction $D$ and obtain
\begin{align*}
D &=  \frac{(n^2-n)!}{n!(n-1)!^n} \le 
\frac%
{ 
\sqrt{2\pi} \cdot e^{\frac1{12n}} \cdot (n^2-n)^{n^2-n+\frac12} \cdot e^{n-n^2}
}
{
 \sqrt{2\pi} \cdot n^{n+\frac12} \cdot e^{-n} \cdot \sqrt{2\pi}^{\,n} \cdot (n-1)^{(n-\frac12)n} \cdot e^{-(n-1)n}
}
\\ &= 
\left(\frac{e}{\sqrt{2\pi}}\right)^{n} \cdot e^{\frac1{12n}} \cdot  
\frac%
{
n^{n^2-n+\frac12} \cdot (n-1)^{n^2-n+\frac12} 
}
{
n^{n+\frac12} \cdot (n-1)^{n^2-\frac{n}2} 
}
= 
\underset{R(n)}{\underbrace{\left( \left( \frac{e}{\sqrt{2\pi}} \right)^{n} \cdot e^{\frac1{12n}}  \cdot  (n-1)^{\frac{1-n}2}\right)}} \cdot n^{n^2-2n}.
\end{align*} 
It is easy to see that $R(n)$ is monotonically decreasing for $n\ge2$ and takes a value smaller than $1$ for $n=3$. Hence, for $n>2$ we have $D< n^{n^2-2n}$. 
\end{proof}

We shall now determine the group $A\subseteq\Z^n $ generated by $S(\Chow_n)$. 
Since $V^n\cq H_n$ is the normalization of $\Chow_n$, 
Proposition~\ref{pro:Ssat-norm} tells us that 
$A$ equals the group generated by the monoid $S(V^n\cq H_n)$. 
The latter is described in terms of plethysms in Lemma~\ref{le:CR-Chow}. 

Recall $V=\C^n$ and $G=\Gl(V)$. 
We say that $\la$ occurs in $\Sym^n\Sym^k V$ if $V_G(\la)$ occurs as a submodule in the latter.  
Recall the convenient notation $k\times d :=(d,\ldots,d,0,\ldots,0)$ for 
a rectangular partition with $k$ rows of length~$d$.

\begin{lemma}\label{le:gen-princ}
If $\la$ occurs in $\Sym^n\Sym^k V$, $\ell\in\N$, then 
$(1\times \ell k) + \la$ 
occurs in $\Sym^{\ell + n}\Sym^k V$.
\end{lemma}

\begin{proof}
Let $V^U=\C v$, i.e., $v$ is a highest weight vector of $V$ with
weight $1\times 1$. 
Then $v^{\ot \ell k} \in \Sym^\ell\Sym^k V$
is a highest weight vector of weight $(1\times \ell k)$.
Let $f\in \Sym^n\Sym^k V$ be a highest weight vector of weight $\la$. 
Then the product $v^{\ot \ell k}\,f\in \Sym^{\ell + n}\Sym^k V$ 
is a highest weight vector of weight $(1\times \ell k) + \la$.
\end{proof}

\begin{lemma}\label{pro:LamdaSym} 
Let $n\ge k\ge 2$, $d:=k(k-1)/2$, and the partition $\mu$ 
of size $k^2$ be obtained by appending to  $2\times d$ a column of length~$k$. 
Further, let $\la$ denote the partition of size $nk$ obtained by appending to $\mu$ 
a row of length~$(n-k)k$. 
Then the partition $\la$ occurs in $\Sym^n \Sym^{k}V$. 
\end{lemma}

\begin{proof}
The $\Gl_2$-module $\Lambda^k\Sym^{k-1}\C^2$ 
is one-dimensional, since $\Sym^{k-1} \C^2$ is of dimension~$k$. 
Hence it contains a nonzero $\Sl_2$-invariant.  
In other words, $2\times d$ occurs in $\Lambda^k \Sym^{k-1}\C^2$.
The ``inheritance principle'' then states that 
$2\times d$ occurs in $\Lambda^k \Sym^{k-1}\C^n$ 
(compare for instance \cite[Lemma 4.3.2]{ike:12b}).   

Cor.~6.4 in~\cite{mami:15} implies that $\mu$ occurs 
in $\Sym^k\Sym^k V$. 
Finally, Lemma~\ref{le:gen-princ} implies the assertion.
\end{proof}

\begin{proposition}\label{pro:AS-Chow}
$S(\Chow_n)$ generates the group 
$\{\la\in\Z^n \mid | \sum_{i=1}^n \la _i \equiv 0 \bmod n \}$ if $n>2$. 
\end{proposition}

\begin{proof}
Using the Schur program\footnote{http://sourceforge.net/projects/schur}
we checked that $(2,2,0\ldots,0)$ occurs in $\Sym^2\Sym^2 V$ 
and $(6,3,0,\ldots,0)$ occurs in $\Sym^3\Sym^3V$. 
Using Lemma~\ref{le:gen-princ}, we conclude that 
$(2n-2,2,0,\ldots,0)$ occurs in $\Sym^n\Sym^2V$, and 
$(3n-3,3,0,\ldots,0)$ occurs in $\Sym^n\Sym^3 V$ if $n\ge 3$. 
(We note that $(3,3,0,\ldots,0)$ does not occur in $\Sym^2\Sym^3V$; 
this is the reason for the assumption $n > 2$.) 
From Lemma~\ref{le:CR-Chow} we conclude that 
$$
 \lambda^{(2)} := (n-1,1,0,\ldots,0) = (3n-3,3,0,\ldots,0) - (2n-2,2,0,\ldots,0)
$$ 
lies in the group~$A$ generated by $S(V^n\cq H_n)$.
Clearly, $\la^{(1)} := (n,0,\ldots,0)\in A$.  

For $3\le k\le n$ let $\la^{(k)}\in A$ denote the partition from Lemma~\ref{pro:LamdaSym}.
Then we have 
$$
 \mbox{$\la^{(k)}_k = 1$ and $\la^{(k)}_i =0$ for $i>k$}
$$
for all $2\le k\le n$. 
This easily implies that 
$\la^{(1)},\ldots,\la^{(n)}$ generate the group 
$L:= \{ \la \in \Z^n \mid \sum_{i=1}^n \la _i \equiv 0 \bmod n\}$. 
Since  $A\subseteq L$ is obvious, we conclude that $A=L$. 
\end{proof}
 
\begin{proof}{(of Theorem~\ref{th:main-Chow})}
This follows from Proposition~\ref{cor:SChow} and 
Proposition~\ref{pro:AS-Chow}, using \eqref{eq:Sat-AC}. 
\end{proof} 

\begin{remark}\label{re:n=2}
The assumption $n>2$ in Theorem~\ref{th:main-Chow} is necessary. 
We have $\Chow_2 = \Sym^2 \C^2$, 
and one can show that $S(\Sym^2 \C^2)$ generates the group 
$\{\la\in\Z^2 \mid \la_1 \equiv \la_2 \equiv 0 \bmod 2 \}$,
compare \cite[\S11.2]{fuha:91}. 
\end{remark}

We conclude by giving an example of an infinite family of holes in the semigroup $S(\Chow_3)$. 

\begin{proposition}
For $j,k \in \N$ let $\la := (7+4k+3j,3+4k,2+4k)$.
Then $\la$ occurs in $\Sym^{4+4k+j} \Sym^3 \C^3$, but not in $S(\Chow_3)$.
\end{proposition}
\begin{proof}
Clearly $(3,0,0)$ occurs in $\Sym^1 \Sym^3 \C^3$.
A calculation with the Schur program reveals that $(4,4,4)$ and $(7,3,2)$ occur in $\Sym^4 \Sym^3 \C^3$.
By the semigroup property $\la$ occurs in $\Sym^{4+4k+j} \Sym^3 \C^3$.
We show that $\la$ does not occur in $\Sym^3 \Sym^{4+4k+j} \C^3$, which means that $\la$ does not occur in $S(V^3 \cq H_3)$ and hence not in $S(\Chow_3)$.
Indeed, applying \cite[Cor.~6.4]{mami:15} twice we can cut away pairs of columns of length 3 as follows.
Cutting away $2k+1$ such pairs we see that
$(7+4k+3j,3+4k,2+4k)$ occurs $\Sym^3 \Sym^{4+4k+j} \C^3$ iff
$(5+3j,1,0)$ occurs $\Sym^3 \Sym^{2+j} \C^3$.
Since this is a nontrivial hook, it does not appear in the plethysm decomposition.
\end{proof}

\noindent We symbolically calculated the holes of $S(\Chow_3)$ whose plethysm coefficient in $\Sym(\Sym^3 \C^3)$
is positive up to degree $9$.
In some cases the plethysm coefficient was $1$, in some cases it was $2$. The cases where it was $2$ are marked with a superscript $2$.
We obtained the following list:
\begin{quotation} \noindent 
$(7,3,2)$, $(10,3,2)$, $(8,4,3)$, $(7,6,2)$, $(13,3,2)$, $(11,4,3)$, $(10,5,3)$, $(9,7,2)$, $(16,3,2)$, $(14,4,3)$, $(13,5,3)$, $(12,5,4)$, $(11,7,3)$, $(10,9,2)$, $(10,6,5)$, $(9,8,4)$, $(19,3,2)$, $(17,4,3)$, $(16,5,3)$, $(15,5,4)$, $(14,7,3)^2$, $(13,9,2)^2$, $(13,6,5)^2$, $(12,7,5)$, $(11,10,3)$, $(11,9,4)$, $(11,7,6)$, $(10,9,5)$, $(22,3,2)$, $(20,4,3)$, $(19,5,3)$, $(18,5,4)$, $(17,7,3)^2$, $(16,6,5)^2$, $(15,7,5)^2$, $(14,7,6)^2$, $(13,12,2)$, $(13,11,3)$, $(13,9,5)^2$, $(12,11,4)$, $(12,8,7)$, $(11,10,6)$.
\end{quotation}

Appendix A.1 of \cite{ike:12b} studies $\Det_3$ and lists partitions up to degree $15$ whose plethysm coefficient exceeds the corresponding symmetric Kronecker coefficient. $191$ of these partitions have a vanishing symmetric Kronecker coefficient. Therefore these are $191$ holes in $S(\Det_3)$ that are not holes in $\Sym(\Sym^3 \C^9)$. This list of holes is not expected to contain all existing holes.



{\small \def\cprime{$'$}

}

\end{document}